\documentclass[runningheads]{llncs}
\usepackage{graphicx}
\graphicspath{{./figs/}}

\usepackage[utf8]{inputenc}
\usepackage{amssymb,amsfonts,amsmath}
\usepackage[inline]{enumitem}
\usepackage{mathtools}
\usepackage[hidelinks]{hyperref}
\usepackage[nameinlink]{cleveref}
\usepackage{csquotes}
\usepackage[labelfont=normalfont]{subfig}
\spnewtheorem{obs}{Observation}{\bfseries}{\itshape}
\usepackage{thmtools, thm-restate}

\crefname{conjecture}{Conjecture}{Conjectures}
\crefname{obs}{Observation}{Observations}
\crefname{theorem}{Theorem}{Theorems}
\crefname{section}{Section}{Sections}
\crefname{appendix}{Appendix}{Appendices}
\crefname{lemma}{Lemma}{Lemmata}
\crefname{figure}{Fig.}{Figs.}
\Crefname{figure}{Figure}{Figures} 
 
\renewcommand{\orcidID}[1]{\href{https://orcid.org/#1}{\includegraphics[scale=.03]{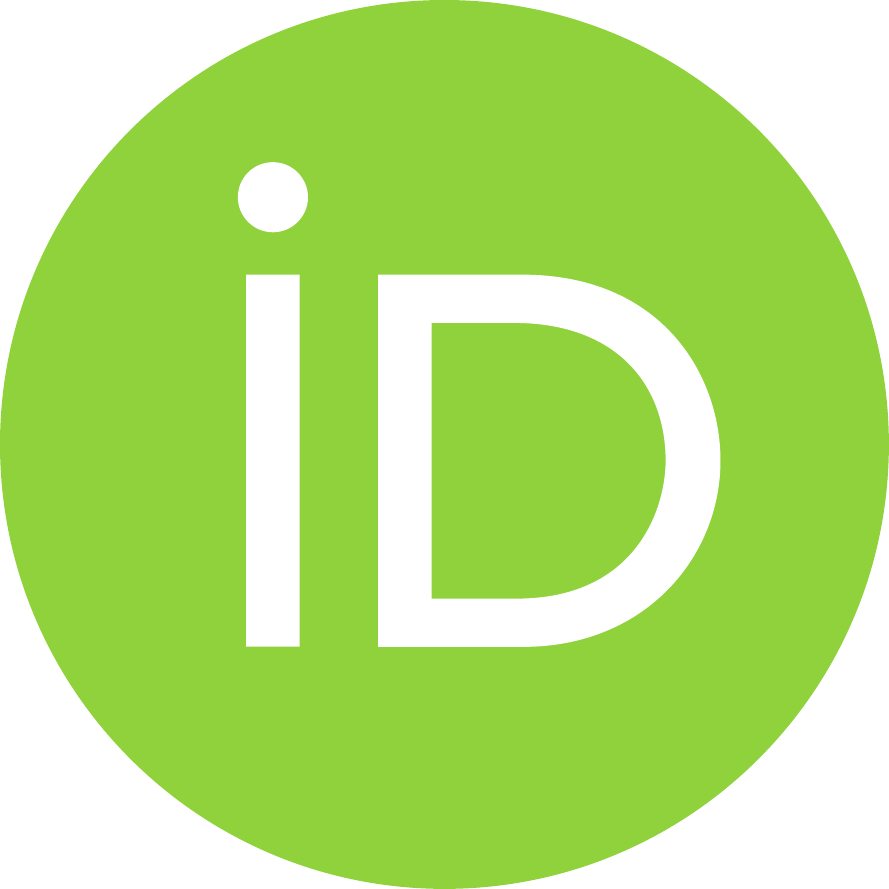}}}

\newenvironment{claimproof}[1]{\underline{Proof of Claim:}\space#1}{\leavevmode\unskip\penalty9999 \hbox{}\nobreak\hfill\quad\hbox{$\blacksquare$}}

\usepackage{xcolor}

\newcommand{\caterpillar}{active subdrawing}

\begin{document}

\title{Plane Spanning Trees in Edge-Colored Simple~Drawings of~$K_{n}$\thanks{
    We are particularly grateful to Irene Parada for bringing this problem to our attention. We also thank the organizers of the $4^{th}$ DACH Workshop on Arrangements, that took place in February 2020 in Malchow and was funded by Deutsche Forschungsgemeinschaft (DFG), the Austrian Science Fund (FWF) and the Swiss National Science Foundation (SNSF). M.~H.~is supported by SNSF Project~200021E-171681. R.~P.~and A.~W.~are supported by FWF grant~W1230. J.~O.~is supported by ERC StG 757609. N.~S.~is supported by DFG Project MU3501/3-1. D.~P. and B.~V. are supported by FWF Project \mbox{I 3340-N35}.
  }}
\author{%
  Oswin~Aichholzer\inst{1}\orcidID{0000-0002-2364-0583} 
  \and
  Michael~Hoffmann\inst{2}\orcidID{0000-0001-5307-7106} 
  \and
  Johannes~Obenaus\inst{3}\orcidID{0000-0002-0179-125X} 
  \and
  Rosna~Paul\inst{1}\orcidID{0000-0002-2458-6427} 
  \and
  Daniel~Perz\inst{1}\orcidID{0000-0002-6557-2355} 
  \and
  Nadja~Seiferth\inst{3}\orcidID{0000-0003-1462-006X} 
  \and
  Birgit~Vogtenhuber\inst{1}\orcidID{0000-0002-7166-4467} 
  \and\\
  Alexandra~Weinberger\inst{1}\orcidID{0000-0001-8553-6661}
}

\authorrunning{Aichholzer, Hoffmann, Obenaus, Paul, Perz, Vogtenhuber, and Weinberger}
\institute{%
  Institute of Software Technology, Graz University of Technology, Graz, Austria\\
  \email{\{oaich,ropaul,daperz,bvogt,aweinber\}@ist.tugraz.at} \and
  Department of Computer Science, ETH Z{\"u}rich, Switzerland\\
  \email{hoffmann@inf.ethz.ch} \and
  Institut f{\"u}r Informatik, Freie Universit{\"a}t Berlin, Germany\\ 
  \email{\{johannes.obenaus,nadja.seiferth\}@fu-berlin.de}
}

\maketitle

  \begin{abstract}
	K{\'{a}}rolyi, Pach, and T{\'{o}}th proved that every 2-edge-colored 
	straight-line drawing of the complete graph contains a monochromatic plane spanning tree. 
	It is open if this statement generalizes to other classes of drawings, specifically, to \emph{simple drawings} of the complete graph. 
	These are drawings where edges are represented by Jordan arcs, any two of which intersect at most once. 
	We present two partial results towards such a generalization. First, we show that the statement holds for cylindrical simple drawings. 
	(In a \emph{cylindrical} drawing, all vertices are placed on two concentric circles and no edge crosses either circle.) 
	Second, we introduce a relaxation of the problem in which the graph is $k$-edge-colored, 
	and the target structure must be \emph{hypochromatic}, that is, avoid (at least) one color class.  
	In this setting, we show that 
	every $\lceil (n+5)/6\rceil$-edge-colored
	monotone simple drawing of $K_n$ contains a hypochromatic plane spanning tree. 
	(In a \emph{monotone} drawing, every edge is represented as an $x$-monotone curve.)

\keywords{Simple drawing \and Cylindrical drawing \and Monotone drawing \and Plane subdrawing.}
\end{abstract}

\section{Introduction}\label{sec:intro}

A \emph{simple drawing} of a graph
represents vertices by pairwise distinct points (in the Euclidean plane) and edges by Jordan 
arcs connecting their endpoints 
such that (1)~no (relative interior of an) edge passes through a vertex and (2) every pair of edges intersect at most once, either in a common endpoint or in their relative interior, forming a proper crossing. 
Simple drawings (also called \emph{good drawings}~\cite{EG_1973} or \emph{simple topological graphs}~\cite{kyncl2009}) have been well studied, amongst others, 
in the context of crossing minimization (see e.g.~\cite{schaefer2017}), as it is known that every crossing-minimal drawing of a graph is simple. 
Also every straight-line drawing is simple. 
Further well-known classes of simple drawings relevant for this work are 
\emph{pseudolinear drawings}, where every edge can be extended to a bi-infinite Jordan arc such that every pair of them intersects exactly once; 
\emph{cylindrical simple drawings}, where all vertices are placed on two concentric circles, no edge crosses either circle, and edges between two vertices on the outer (inner) circle lie completely outside (inside) that circle;  
\emph{2-page book drawings}, where all vertices lie on a line and no edge crosses that line; and 
\emph{monotone simple drawings}, where all edges are \mbox{$x$-monotone} curves. 
Unless explicitly mentioned otherwise, all considered drawings are simple, and the term simple is mostly omitted.

In this paper we are concerned with finding plane substructures in simple drawings. Specifically, we study the existence of plane spanning trees in edge-colored simple drawings of the complete graph $K_n$. A \emph{$k$-edge-coloring} of a graph is a map from its edge set to a set of $k$ colors.\footnote{Note that the coloring need \emph{not} be proper nor have any other special properties.} 
A subgraph $H$ of a $k$-edge-colored graph $G$ is \emph{hypochromatic} if the edges of $H$ use at most $k-1$ colors, that is, $H$ avoids at least one of the $k$ color classes. If all edges of $H$ have the same color, then $H$ is \emph{monochromatic}. We are inspired by the following conjecture.

\begin{conjecture}\label{conj:main}
  Every 2-edge-colored simple drawing of $K_n$ contains a monochromatic plane spanning tree.
\end{conjecture}

K{\'{a}}rolyi, Pach, and T{\'{o}}th~\cite{straight_lines} proved the statement for straight-line drawings, where the $2$-edge-coloring can also be interpreted as a Ramsey-type setting, where one color corresponds to the edges of the graph and the other color to the edges of its complement. Such an interpretation is less natural in the topological setting, where the edges are not implicitly defined by placing the vertices.

Unfortunately, a proof of \cref{conj:main} seems elusive. 
However, we show that it holds for specific classes of simple drawings, such as $2$-page book drawings, pseudolinear drawings, and cylindrical drawings.
The result for $2$-page book drawings can be shown straightforwardly. 
The statement for pseudolinear drawings follows from generalizing the proof for straight-line drawings by K{\'{a}}rolyi, Pach, and T{\'{o}}th~\cite{straight_lines} to this setting. 

\begin{restatable}{proposition}{thmTwoPageBook}\label{prop:2-page-book}
  Every $2$-edge-colored $2$-page book drawing of~$K_n$ contains a plane monochromatic spanning tree. 
\end{restatable}
\begin{restatable}{proposition}{thmPseudolinear}\label{prop:pseudolinear}
	Every 2-edge-colored pseudolinear drawing of~$K_n$ contains a plane monochromatic spanning tree.
\end{restatable}
See \cref{app:preliminary_results} for proofs of those statements. \newpage
The result for cylindrical drawings is more involved; 
it forms our first main contribution. 

\begin{restatable}{theorem}{thmCylindrical}\label{thm:cylindrical}
  Every $2$-edge-colored cylindrical simple drawing of~$K_n$ contains a monochromatic plane spanning tree.
\end{restatable}

In light of the apparent challenge in attacking \cref{conj:main}, we also consider the following generalized formulation, which uses more colors.

\begin{conjecture}\label{conj:gen}
  For $k\ge 2$, every $k$-edge-colored simple drawing of $K_n$ contains a hypochromatic plane spanning tree.
\end{conjecture}

Note that both conjectures are in fact equivalent: On the one hand, \cref{conj:gen}
implies \cref{conj:main} by setting $k=2$. On the other hand, assuming \cref{conj:gen} holds for some $k$, it also holds for every larger $k'$ because we can simply merge color classes until we are down to $k$ colors. Avoiding any one of the resulting color classes also avoids at least one of the original color classes.

Our second result is the following statement about \emph{monotone} drawings. 

\begin{theorem}\label{thm:hypon2}
  Every $\lceil(n+5)/6\rceil$-edge-colored monotone simple drawing of~$K_n$ contains a hypochromatic plane spanning tree.  \end{theorem}

Finally, note that \emph{some} assumptions concerning the drawing are necessary to obtain any result on the existence of plane substructures. Without any restriction, every pair of edges may cross. The class of simple drawings is formed by two restrictions: forbid adjacent edges to cross and forbid independent edges to cross more than once. 
Both restrictions are necessary in the statement of \cref{conj:main}. 
If adjacent edges
may cross, then one can construct drawings where every pair of adjacent edges crosses (e.g., in the neighborhood of the common vertex), implying that no plane substructure can have a vertex of degree more than one.  
And for \emph{star-simple} drawings, where adjacent edges do not cross but independent edges may cross more than once, already $K_5$ admits $2$-edge-colored star-simple
drawings without any monochromatic plane spanning tree; see \cref{fig:starsimple_noplanemonochromatictree}.

\begin{figure}[htb]
	\centering
	\includegraphics[scale=0.6,page=1]{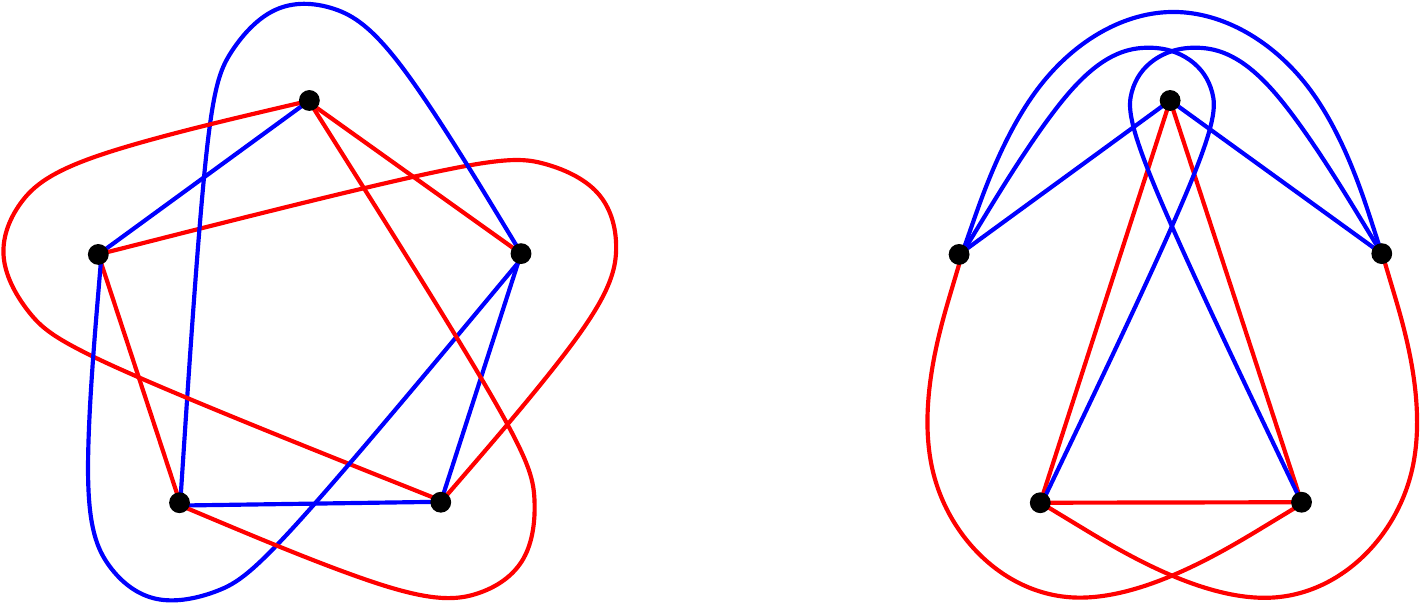}
	\caption{
          Star-simple drawings of $K_5$ without monochromatic plane spanning tree.}
	\label{fig:starsimple_noplanemonochromatictree}
\end{figure}

\paragraph{Related Work.} The problem of finding plane subdrawings in a given drawing has gained some attention over the past decades.
We mention only a few results from the vast literature on plane substructures. In 1988, Rafla~\cite{rafla} conjectured that every simple drawing of $K_n$ contains a plane Hamiltonian cycle. By now the conjecture is known to be true for~$n \leq 9$~\cite{aafhpprsv-agdsc-15} and several classes of simple drawings (e.g., 2-page book drawings, monotone drawings, cylindrical drawings), but remains open in general. See also \cite{bernhart1979book,biniaz_garcia,bose2006partitions,keller2013blockers,rivera2013sufficient} for some results about plane spanning trees in straight-line drawings of complete graphs. In an edge-colored setting, many other coloring schemes were studied in this context, see 
e.g.~\cite{brualdi1996multicolored,erdos1983some}.
 
Observe that if one color class of a drawing is not spanning, the drawing of the remaining colors contains a complete bipartite graph as a subdrawing. 
Recently, it has been shown that every simple drawing of the complete bipartite graph contains a plane spanning tree~\cite{complete_biparite_spanning}.
Consequently, this implies the following lemma, which turns out to be useful later on (see \cref{app:preliminary_results} for the proof).
\begin{restatable}{lemma}{obsNonSpanning}\label{obs:span}
  Let~$D$ be a $k$-edge-colored simple drawing of~$K_n$, for $k\ge 2$. If one of the color classes is not spanning, then $D$ contains a hypochromatic plane spanning tree.
\end{restatable}

\section{Cylindrical Drawings}\label{sec:special_mono}

This section is devoted to
\cref{thm:cylindrical}, which states that every $2$-edge-colored cylindrical drawing of~$K_n$ contains a monochromatic plane spanning tree. 
We give a detailed outline of the proof. 
The full proof can be found in \cref{app:full_proof_cylindrical}.

For easier readability, we introduce some 
names for the different elements of a cylindrical drawing (cf.\ \cref{fig:cylindrical_notation}).
We call the vertices on the inner (outer) circle \emph{inner (outer) vertices}. 
Similarly, we call edges connecting two inner (outer) vertices \emph{inner (outer) edges}; the remaining edges are called \emph{side edges}. 
The edges between consecutive vertices on the inner (outer) circle are called \emph{cycle edges} and
the union of all inner (outer) cycle edges are called inner (outer) \emph{cycle}.
The definition of cylindrical drawings implies that all cycle edges are uncrossed.
The \emph{rotation} of a vertex $v$ is
the circular ordering of all edges incident to $v$.
In this ordering, the
cycle edges separate the inner (outer) edges from the side edges. 
Hence, the rotation of $v$ induces a linear order on the side edges incident to $v$.

\begin{figure}[htb]
	\centering
	\includegraphics[scale=0.6,page=1]{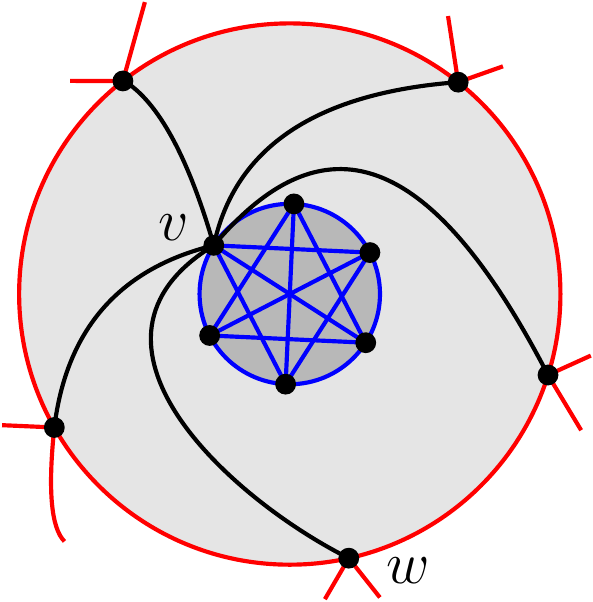}
	\caption{%
		Sketch of a cylindrical drawing. 
		Inner
		edges are drawn blue, outer  
		edges red, and side edges black.
		$vw$ is the first side edge in the clockwise rotation of $v$.
		  }
	\label{fig:cylindrical_notation}
\end{figure}

\begin{proof}[sketch]
Our proof
consists of two steps. 
In Step~1, we restrict considerations to drawings fulfilling two properties, for which we compute a monochromatic plane spanning subgraph using a multi-stage sweep algorithm. 
In Step~2, we show how to handle drawings that do not fulfill all properties from Step~1.

\paragraph{Step 1.} Let $D$ be a 2-edge-colored cylindrical drawing that fulfills the following properties:
\begin{enumerate}[leftmargin=*,label={(P\arabic*)}]
\item\label{p:1} $D$ has inner and outer vertices, and
\item\label{p:34} $D$'s inner and outer cycle are both monochromatic, but of different color.
\end{enumerate} 

Assume without loss of generality that the inner cycle of $D$ is blue and hence the outer cycle is red.
We will refer to them as the blue and red cycle and to the vertices on them as blue and red vertices, respectively.

We use the following algorithm to compute a (bichromatic)
subdrawing~$H$ of $D$ consisting of some side edges of $D$ and their endpoints 
	(cf.\ \cref{fig:cylindrical_algo2}).  
\begin{enumerate}[leftmargin=1.7em,itemindent=3em,start=0,label=\textbf{Phase~{\arabic*}.},ref=Phase~{\arabic*}]	
	\item\label{a:1} Initially, let $H$ be empty. 
	Choose an arbitrary inner vertex 
		as initial \emph{rotation vertex} $v_{\mathrm{cur}}$, 
		set the \emph{rotation direction} to clockwise, and
	set the first side edge of $v_{\mathrm{cur}}$ in the rotation direction as initial \emph{current edge} $e_{\mathrm{cur}}$. 
	\item\label{a:2} 
		We repeat the following process while $e_{\mathrm{cur}}$ is a side edge and  
			while $H$ is still missing vertices from the cycle of $D$ not containing $v_{\mathrm{cur}}$:
			Add $e_{\mathrm{cur}}$ to $H$;
			If $e_{\mathrm{cur}}$ does not have the same color as $v_{\mathrm{cur}}$, set $v_{\mathrm{cur}}$ to be the other endpoint of $e_{\mathrm{cur}}$ and reverse the rotation direction (clockwise $\leftrightarrow$ counterclockwise);
			In any case, set $e_{\mathrm{cur}}$ to be the next edge incident to $v_{\mathrm{cur}}$ after $e_{\mathrm{cur}}$ in the (possibly changed) rotation direction.
	\item\label{a:3} If $H$ contains all vertices of $D$ from the cycle not containing $v_{\mathrm{cur}}$: Return~$H$.
	\item\label{a:4} Otherwise: Set $H_{\mathrm{prev}}=H$, reset $H$ to be empty,  
		reverse 
		the rotation direction,  
		set $e_{\mathrm{cur}}$ to be the first side edge of $v_{\mathrm{cur}}$ in the new rotation direction, and restart with \ref{a:2}.
\end{enumerate} 

\begin{figure}[htb]
	\centering
	\includegraphics[scale=0.6,page=1]{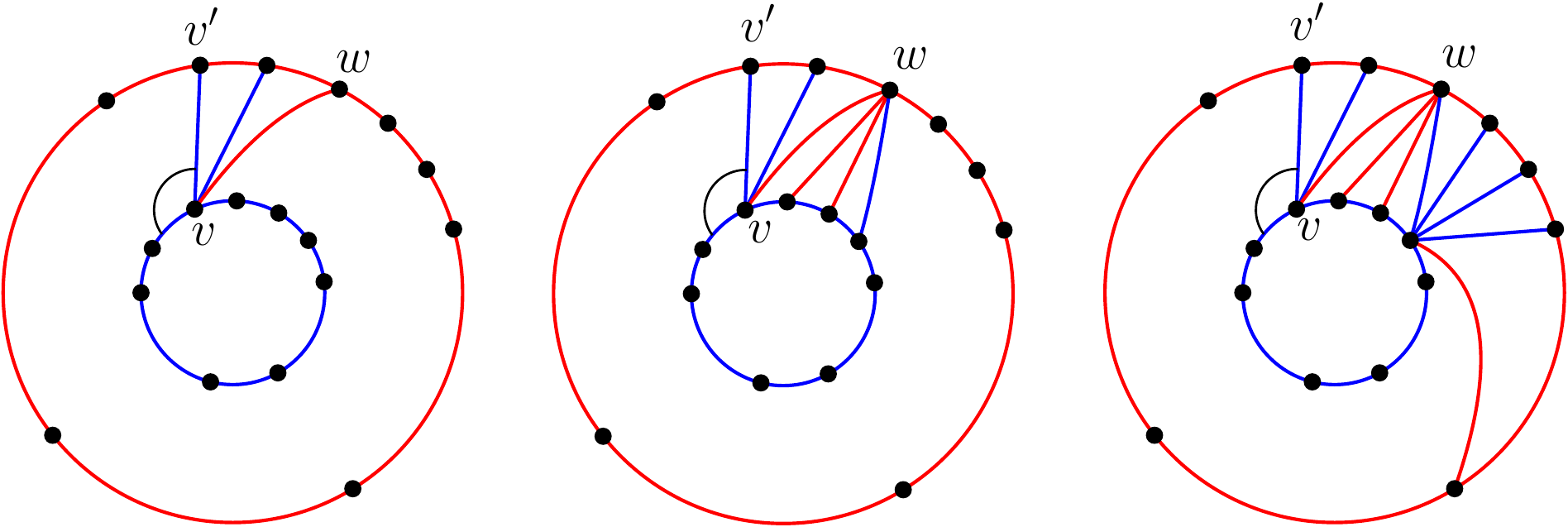}
	\caption{%
		The first steps of the algorithm. The black arc at vertex~$v$ indicates that $vv'$ is the first side edge of $v$ in clockwise order (the initial rotation direction).}
	\label{fig:cylindrical_algo2}
\end{figure}

The following invariants hold for the algorithm (see \cref{app:full_proof_cylindrical} for a proof):
\begin{enumerate}[leftmargin=*,label={(J\arabic*)}]
\item\label{j:1} At any time, the union of $H$ and the two cycles of $D$ forms a plane drawing. 
\item\label{j:2} Any blue (red) vertex in $H$ is incident to a red (blue) edge in $H$, except for the current rotation vertex. 
\item\label{j:4} Assume that \ref{a:2} is performed more than once and let $V(H)$ be the set of vertices of $H$. 
  Then for any $i\geq 2$, after round $i$ of \ref{a:2}, either $V(H)$ is a strict superset of $V(H_{\text{prev}})$ or $H$ contains all vertices from the cycle not containing $v_{\mathrm{cur}}$, the current rotation vertex (or both conditions hold).
\end{enumerate}

Using those invariants, we can now complete Step 1: 
By \ref{j:4}, the algorithm  
terminates. 
And by \ref{j:1} and \ref{j:2}, at least one of the color classes of the union of $H$ and the two cycles of $D$ is a monochromatic plane spanning graph for $D$. 

\paragraph{Step 2.} 
Now assume that $D$ violates
at least one of the properties \ref{p:1} and \ref{p:34}.

If it violates \ref{p:1}, then $D$ is isomorphic to a 2-page book drawing and hence contains a monochromatic plane spanning tree (see Proposition~\ref{prop:2-page-book}).

If $D$  
does not fulfill \ref{p:34}, then we remove vertices whose cycle edges are of different color until we 
reach a subdrawing $D'$ where both cycles are monochromatic, find a plane monochromatic spanning tree on $D'$ by either Step 1 or \cref{obs:span}, and then extend it to a 
monochromatic spanning tree on $D$.~\hfill$\qed$
\end{proof}

\section{Monotone Drawings}\label{sec:n:2_colors_mono}

In this section, we prove the existence of hypochromatic plane
spanning trees in $k$-edge-colored monotone drawings of~$K_n$, 
for $k$ linear in $n$.
\begin{lemma}\label{lem:computation}
Conjecture~\ref{conj:main} holds for any simple drawing of $K_n$ with $n \leq 7$ vertices.
\end{lemma}
For $n \leq 4$ this can easily be observed by hand.
For $n=5,\ldots ,7$ we considered all weak isomorphism classes\footnote{Two simple drawings of $K_n$ are weakly isomorphic iff they have the same crossing edge pairs.} of simple drawings of $K_n$~\cite{aafhpprsv-agdsc-15} and checked for all possible 2-edge colorings that there exists a monochromatic plane spanning tree.
Computations for $n=8$ are currently out of reach, 
as there are 5,370,725 weak isomorphism classes of simple drawings~\cite{aafhpprsv-agdsc-15} and more than $10^8$ possible colorings for each of them.

\begin{proof}[of \cref{thm:hypon2}]
  Let $d\ge 2$ be an integer constant, and let
  $k=\lceil(n+d-1)/d\rceil=\lceil(n-1)/d\rceil+1$. The argument works
  for any $d$ so that
  \cref{conj:main} holds for all monotone drawings on up to $d+1$
  vertices.

  Consider a $k$-edge-colored monotone drawing
  $D$ of~$K_n$, and let $v_0, v_1, \ldots, v_{n-1}$ denote the sequence
  of
  vertices in increasing $x$-order. We partition the vertices into
  $k-1$ groups $G_0,\ldots,G_{k-2}$ of size at most $d+1$ by
  setting $G_i=(v_{di},v_{di+1},\ldots,v_{di+d})$. (The last group may
  have less than $d+1$ vertices.) Observe that
  $G_i\cap G_{i+1}=\{v_{d(i+1)}\}$.

  We proceed in two phases. In both phases we consider each group
  separately. At the end of the first phase, we choose which color to
  remove. At the end of the second phase, we have an induced plane
  spanning tree $T_i$ for $G_i$ that avoids the chosen color, for each
  $i\in\{0,\ldots k-2\}$. As $D$ is monotone, the union
  $\bigcup_{i=0}^{k-2}T_i$ forms a hypochromatic plane
  spanning tree in $D$.

  In the first phase, we consider each group $G_i$, and check whether
  it has a monochromatic plane spanning tree in some color $c$. If so,
  we put $c$ in a set $S$ of colors to keep. If not, then by
  \cref{conj:main} (which we assume to hold for $G_i$, as $G_i$ has at
  most $d+1$ vertices) we can remove \emph{any} single color and still
  find a monochromatic plane spanning tree in $G_i$. (If $c$ is the
  color to be removed, then consider the bicoloring where all colors
  other than $c$ are merged into a single second color.) As
  $|S|\le k-1$, we can choose a color not in $S$ to be removed
  at the end of the first phase.

  In the second phase, for each group $G_i$ we either select a
  monochromatic plane spanning tree (if it exists), or find a plane
  spanning tree that avoids the chosen color.
  
  To obtain the statement of \cref{thm:hypon2}, we use the result of Lemma~\ref{lem:computation}.~\hfill$\qed$
\end{proof}

\section{Open Problems}\label{sec:open}

Besides resolving the conjectures in full generality, it would be interesting to prove them for other specific classes of drawings (e.g., monotone).  A useful step in this direction would be to expand the range of $k$ for which \cref{conj:gen} holds.

\bibliography{plane_trees_colored_edges}

\begin{thebibliography}{10}
\providecommand{\url}[1]{\texttt{#1}}
\providecommand{\urlprefix}{URL }
\providecommand{\doi}[1]{https://doi.org/#1}

\bibitem{aafhpprsv-agdsc-15}
\'Abrego, B., Aichholzer, O., Fern\'andez-Merchant, S., Hackl, T., Pammer, J.,
  Pilz, A., Ramos, P., Salazar, G., Vogtenhuber, B.: All good drawings of small
  complete graphs. In: Abstracts $31^{st}$ European Workshop on Computational
  Geometry (EuroCG'15). pp. 57--60 (2015)

\bibitem{complete_biparite_spanning}
Aichholzer, O., Garc\'ia, A., Parada, I., Vogtenhuber, B., Weinberger, A.:
  Simple drawings of~{$K_{m,n}$} contain shooting stars. In: Abstracts
  $35^{th}$ European Workshop on Computational Geometry (EuroCG'20). pp.
  36:1--36:7 (2020)

\bibitem{bernhart1979book}
Bernhart, F., Kainen, P.C.: The book thickness of a graph. Journal of
  Combinatorial Theory, Series B  \textbf{27}(3),  320--331 (1979).
  \doi{10.1016/0095-8956(79)90021-2}

\bibitem{biniaz_garcia}
Biniaz, A., Garc{\'{\i}}a, A.: Partitions of complete geometric graphs into
  plane trees. Computational Geometry  \textbf{90},  101653 (2020).
  \doi{10.1016/j.comgeo.2020.101653}

\bibitem{bose2006partitions}
Bose, P., Hurtado, F., Rivera-Campo, E., Wood, D.R.: Partitions of complete
  geometric graphs into plane trees. Computational Geometry  \textbf{34}(2),
  116--125 (2006). \doi{10.1016/j.comgeo.2005.08.006}

\bibitem{brualdi1996multicolored}
Brualdi, R.A., Hollingsworth, S.: Multicolored trees in complete graphs.
  Journal of Combinatorial Theory, Series B  \textbf{68}(2),  310--313 (1996).
  \doi{10.1006/jctb.1996.0071}

\bibitem{EG_1973}
Erd\H{o}s, P., Guy, R.: Crossing number problems. The American Mathematical
  Monthly  \textbf{88},  52--58 (1973)

\bibitem{erdos1983some}
Erd{\H o}s, P., Ne{\v{s}}etril, J., R{\"o}dl, V.: Some problems related to
  partitions of edges of a graph. Graphs and other combinatorial topics,
  Teubner, Leipzig  \textbf{5463} (1983)

\bibitem{euclidean_pseudo}
Goodman, J.E.: Proof of a conjecture of {Burr}, {Gr\"unbaum}, and {Sloane}.
  Discrete Mathematics  \textbf{32}(1),  27--35 (1980).
  \doi{10.1016/0012-365X(80)90096-5}

\bibitem{straight_lines}
K{\'{a}}rolyi, G., Pach, J., T{\'{o}}th, G.: Ramsey-type results for geometric
  graphs, {I}. Discrete {\&} Computational Geometry  \textbf{18},  247--255
  (1997). \doi{10.1007/PL00009317}

\bibitem{keller2013blockers}
Keller, C., Perles, M.A., Rivera-Campo, E., Urrutia-Galicia, V.: Blockers for
  noncrossing spanning trees in complete geometric graphs. In: Thirty Essays on
  Geometric Graph Theory, pp. 383--397. Springer (2013).
  \doi{10.1007/978-1-4614-0110-0\_20}

\bibitem{kyncl2009}
Kyn\v{c}l, J.: Enumeration of simple complete topological graphs. European
  Journal of Combinatorics  \textbf{30},  1676--1685 (2009).
  \doi{10.1016/j.ejc.2009.03.005}

\bibitem{rafla}
Rafla, N.H.: The good drawings {$D_n$} of the complete graph {$K_n$}. Ph.D.
  thesis, McGill University, Montreal (1988)

\bibitem{rivera2013sufficient}
Rivera-Campo, E., Urrutia-Galicia, V.: A sufficient condition for the existence
  of plane spanning trees on geometric graphs. Computational Geometry
  \textbf{46}(1), ~1--6 (2013). \doi{10.1016/j.comgeo.2012.02.006}

\bibitem{schaefer2017}
Schaefer, M.: The graph crossing number and its variants: A survey. Electronic
  Journal of Combinatorics, Dynamic Survey  \textbf{21}(4) (2020).
  \doi{10.37236/2713}

\end{thebibliography}

\appendix

\section{Preliminary Results}\label{app:preliminary_results}

\begin{obs}\label{obs:2_uncrossed}
Let $D$ be a 2-edge-colored simple drawing of $K_n$ and $v$ be a vertex incident to (monochromatically) uncrossed edges of both color classes. If $D \setminus \{ v\}$ contains a monochromatic plane spanning tree, then so does $D$.
\end{obs}

Using this observation, it is not hard to see that any 2-page book drawing contains a monochromatic plane spanning tree.

\thmTwoPageBook*

\begin{proof}
	Let $D$ be a 2-edge-colored $2$-page book drawing of~$K_n$ (with colors red and blue). As long as $D$ contains vertices that are incident to at least one (monochromatically) uncrossed blue edge and at least one (monochromatically) uncrossed red edge, iteratively remove these vertices to obtain a subdrawing $D'$. Clearly, $D'$ remains a $2$-page book drawing.
	
	Label the vertices in $D'$ with $v_1, v_2, .., v_j$ from left to right. By the properties of $2$-page book drawings, the edges $v_{i}v_{i+1}$ between consecutive vertices are uncrossed. Since no vertex is adjacent to two differently colored, uncrossed edges, the path $v_1\ldots v_j$ is monochromatic and thus forms a plane monochromatic spanning tree for $D'$.
	
	Finally, using \cref{obs:2_uncrossed}, re-add the previously removed vertices in inverse order to obtain a plane monochromatic spanning tree of $D$.~\hfill$\qed$
\end{proof}

\thmPseudolinear*

The proof for the straight line case in~\cite{straight_lines} uses two concepts: 
The existence of a convex hull and the monotonicity of all edges. 
We will observe that pseudolinear drawings fulfill both and then follow the lines of the straight-line proof.

\begin{proof}
	The proof goes by induction on the number $n$ of vertices. 
	As induction base let~$n=2$. 
	Then there is a plane monochromatic spanning tree consisting of the only edge in the drawing. 
	So assume that any pseudolinear drawing of~$K_{n-1}$ contains a plane monochromatic spanning tree. 
	
	For the induction step we consider a $2$-colored pseudolinear drawing of~$K_n$ and call it~$D$. 
	We consider first the case that there exists a vertex~$v$ that is incident to an uncrossed red and an uncrossed blue edge. 
	Then the subdrawing $D\setminus \{v\}$ contains a plane monochromatic spanning tree by our induction hypothesis.  
	Thus, the drawing~$D$ contains a plane monochromatic spanning tree by \cref{obs:2_uncrossed}. 
	
	So let us assume that~$D$ does not contain any vertex that is incident to two differently colored crossing-free edges. 
	To prove that~$D$ contains a plane monochromatic spanning tree in this case, we will use 
	the following well known fact (whose proof we include for the sake of self-containment).
	
	\begin{claim}
		The outermost edges of any pseudolinear drawing of~$K_n$ form an uncrossed cycle.
	\end{claim}

	\begin{claimproof}
		Assume, for contradiction, there are edges~$e_1=u_1v_1$,~$e_2=u_2v_2$ that lie partly on the boundary and cross each other. 
		Let~$H$ be the subdrawing induced by $\{u_1,u_2,v_1,v_2\}$. 
		Let~$e_1$ be extended by pseudoline~$\ell_1$ and~$e_2$ be extended by~$\ell_2$. 
		The pseudolines~$\ell_1$ and~$\ell_2$ are intersected by all edges of~$H \setminus \{e_1,e_2\}$ in one of the vertices of~$H$. 
		Thus, they cannot have a crossing point with any edges of~$H \setminus \{e_1,e_2\}$ in the interior of that edge.  
		Thus, every edge other than~$e_1$ and~$e_2$ has to stay completely on one side of each of $\ell_1$ and $\ell_2$, respectively. 
		The edges together form a cycle that completely encloses~$e_1$ and~$e_2$; see Figure~\ref{fig:pseudoline}. 
		This means in particular that neither~$e_1$ nor~$e_2$ can lie (partly) on the boundary of the drawing.
		\begin{figure}[htb]
			\centering
			\includegraphics[scale=0.6,page=1]{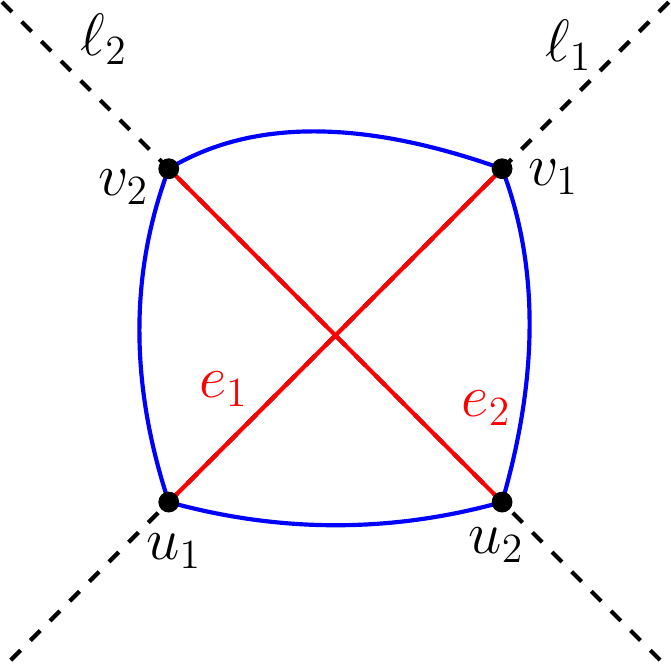}
			\caption{The edges~$e_1$ and~$e_2$ cross and are extended by the (black, dashed) pseudolines~$\ell_1$ and~$\ell_2$. 
			The (blue) edges that are in the same~$K_4$ are forced by~$\ell_1$ and~$\ell_2$ to stay on one side of the crossing.}
			\label{fig:pseudoline}
		\end{figure}	
	\end{claimproof}
	\vspace{2ex}
	
	Since~$D$ does not contain any vertices that are incident to crossing-free edges colored in different colors, 
	it follows that the boundary cycle of~$D$ is monochromatic.
	Assume without loss of generality that the boundary cycle of~$D$ is red. 
	If all vertices lie on the boundary, the boundary edges form a plane red spanning tree and we are done.
	Otherwise there exists at least one interior vertex. 
	Since by~\cite{euclidean_pseudo}, every pseudoline arrangement is isomorphic to a pseudoline arrangement in which every pseudoline is $x$-monotone,
	we can assume that our pseudolinear drawing is $x$-monotone. 
	This implies that there are at least two more uncrossed edges:
	One uncrossed edge is incident to the leftmost vertex and the leftmost vertex that is not on the boundary; 
	another uncrossed edge is incident to the rightmost vertex and the rightmost vertex that is not on the boundary. 
	Both edges have to be red, because~$D$ does not contain any vertices that are adjacent to two differently colored uncrossed edges.

	By the assumption that our pseudolinear drawing is $x$-monotone we can label the vertices~$x_1, x_2,...,x_n$ in $x$-monotone order. 
	By our induction hypothesis, the subdrawings induced by~$x_1, x_2,...,x_i$ and by~$x_i,x_{i+1},....,x_n$ contain 
	plane monochromatic spanning trees for any $i\in \{2,...,n-1\}$. 
	Let~$T_i^l$ be the plane monochromatic spanning tree of the subdrawing induced by~$x_1, x_2,...,x_i$ 
	and~$T_i^r$ the plane monochromatic spanning tree subdrawing induced by~$x_i,x_{i+1}...,x_n$ . 
	If both of them have the same color, then~$T_i^l \bigcup T_i^r$ forms a plane monochromatic spanning tree for the whole drawing. 
	So assume that they have different colors.
	
	We know from the color of the first and the last edge that $T_1^l$ and $T_{n-1}^r$ are red. 
	Thus there has to be an~$i$ for which~$T_i^l$ is red and $T_{i+1}^r$ is red as well. 
	If the edge~$x_ix_{i+1}$ is red, we can use it to connect the two spanning trees. 
	If the edge is blue, it is not part of the boundary cycle. 
	We can use the boundary edge above or the boundary edge below~$x_ix_{i+1}$ to connect the two spanning trees.~\hfill$\qed$
\end{proof}

\obsNonSpanning*

\begin{proof} 
	Assume, without loss of generality, that the edges of the red color class contain no spanning tree (not even a crossing one). 
	Then the subdrawing induced by the red edges has at least two different components. 
	Let~$A$ be the vertex set of one of those components and let~$B$ be the vertices that are not in~$A$. 
	There are no red edges between~$A$ and~$B$. 
	This means that the subdrawing induced by the remaining edges contains a complete bipartite graph with sides of the partition~$A$ and~$B$. 
	Every complete bipartite graph contains a plane spanning tree~\cite{complete_biparite_spanning}. 
	Thus~$D$ contains a plane hypochromatic spanning tree (consisting of only non-red edges).~\hfill$\qed$
\end{proof}

\section{Full proof of \cref{thm:cylindrical}}\label{app:full_proof_cylindrical}

\thmCylindrical*

\begin{proof}
Our proof consists of two steps. 
In Step~1, we restrict considerations to drawings fulfilling two properties, for which we compute a monochromatic plane spanning subdrawing using a multi-stage sweep algorithm.
In Step~2, we show how to handle drawings that do not fulfill all properties from Step~1.

\paragraph{Step 1.} Let $D$ be a 2-edge-colored cylindrical drawing that fulfills the following properties:
\begin{enumerate}[leftmargin=*,label={(P\arabic*)}]
	\item[\ref{p:1}] $D$ has inner and outer vertices, and
	\item[\ref{p:34}] $D$'s inner and outer cycle are both monochromatic, but of different color.
\end{enumerate} 

To simplify the description, we assume without loss of generality that the inner cycle of $D$ is blue and hence the outer cycle is red.
We will refer to them as the blue and red cycle and to the vertices on them as blue and red vertices, respectively. 
We remark that, if there are less than three vertices on a cycle, then the cycle is in fact not a cycle in the graph-theoretic meaning, as it has at most one edge. Moreover, if a cycle has only one vertex and hence does not have any edges, we can assume it to be of any color.

We use the following algorithm to compute a (possibly bichromatic) 
subdrawing~$H$ of $D$ consisting of a subset of side edges of $D$ and their endpoints (cf.~\cref{fig:cylindrical_algo}). 

\begin{enumerate}[leftmargin=1.7em,itemindent=3em,start=0,label=Phase~{\arabic*}.,ref=Phase~{\arabic*}]
	\item[\textbf{\ref{a:1}.}] 
		Initially, let $H$ be empty. 
		Choose an arbitrary inner vertex 
		as initial \emph{rotation vertex} $v_{\mathrm{cur}}$, 
		set the \emph{rotation direction} to clockwise, and
	set the first side edge of $v_{\mathrm{cur}}$ in the rotation direction as initial \emph{current edge} $e_{\mathrm{cur}}$. 
\item[\textbf{\ref{a:2}.}] 
		We repeat the following process while $e_{\mathrm{cur}}$ is a side edge and 
			while $H$ is still missing vertices from the cycle of $D$ not containing $v_{\mathrm{cur}}$:
			Add $e_{\mathrm{cur}}$ to $H$;
			If $e_{\mathrm{cur}}$ does not have the same color as $v_{\mathrm{cur}}$, set $v_{\mathrm{cur}}$ to be the other endpoint of $e_{\mathrm{cur}}$ and reverse the rotation direction (clockwise $\leftrightarrow$ counterclockwise);
			In any case, set $e_{\mathrm{cur}}$ to be the next edge incident to $v_{\mathrm{cur}}$ after $e_{\mathrm{cur}}$ in the (possibly changed) rotation direction.		
	\item[\textbf{\ref{a:3}.}] If $H$ contains all vertices of $D$ from the cycle not containing $v_{\mathrm{cur}}$: Return~$H$. 
	\item[\textbf{\ref{a:4}.}] Otherwise: Set $H_{\mathrm{prev}}=H$, reset $H$ to be empty, 
		reverse the rotation direction,
		set $e_{\mathrm{cur}}$ to be the first side edge of $v_{\mathrm{cur}}$ in the new rotation direction, and restart with \ref{a:2}.
\end{enumerate}

 \begin{figure}[htb]
	\centering
	\includegraphics[scale=0.6,page=1]{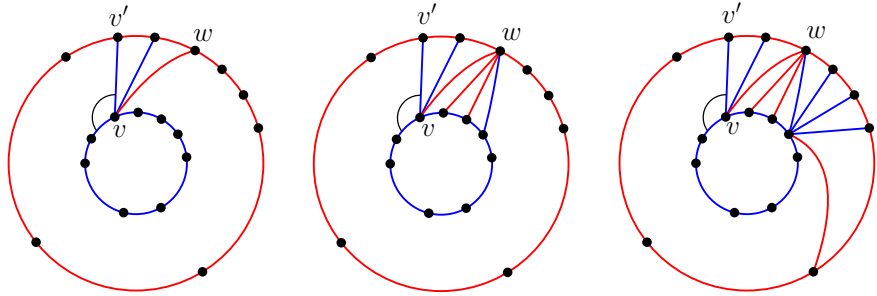}
	\caption{
	 The first steps of our algorithm. The black arc at vertex~$v$ indicates that $vv'$ is the first side edge of $v$ in clockwise order (the first rotation direction). }
	\label{fig:cylindrical_algo}
\end{figure}

Intuitively speaking, this algorithm sweeps back and forth in a zig-zag manner (see \figurename~\ref{fig:cylindrical_algo} for an illustration). We remark that \ref{a:2} adds at least one edge to $H$, namely $e_{\mathrm{cur}} = vv'$ as set in \ref{a:1}. Moreover, the \emph{\caterpillar}~$H$ constructed in \ref{a:2} of the algorithm consists of a main path (also called \emph{backbone path}) of alternating red and blue edges corresponding to the switches between the two cycles, i.e., each vertex along the backbone path has been a rotation vertex. Additionally, each vertex of this backbone path may have an arbitrary number of monochromatic leaves attached. This graph structure is called \emph{caterpillar}. 

An illustration of \ref{a:2} in reverse direction can be found in \figurename~\ref{fig:cylindrical_revert}. As we will see later, at least some edges causing a switch of cycles will differ from the previous backbone edges.
However, the first edges in the reverse process (until the first switch) are the same as the last edges of the previous iteration.

Of course, the graph $H$ (returned in \ref{a:3}) is not the plane monochromatic spanning tree we are looking for. But we claim that either the red cycle together with the red edges of $H$ or the blue cycle together with the blue edges of $H$ forms a plane monochromatic spanning subdrawing of $D$. 

To prove this and thereby the correctness of our algorithm, we need the following invariants concerning the \caterpillar\ $H$.

\begin{enumerate}[leftmargin=*,label={(J\arabic*)}]
\item[\ref{j:1}] At any time, the union of $H$ and the two cycles of $D$ forms a plane drawing. 
\item[\ref{j:2}] Any blue (red) vertex in $H$ is incident to a red (blue) edge in $H$, except for the current rotation vertex.
\item[\ref{j:4}] Assume that \ref{a:2} is performed more than once and let $V(H)$ be the set of vertices of $H$. 
  Then for any $i\geq 2$, after round $i$ of \ref{a:2}, either $V(H)$ is a strict superset of $V(H_{\text{prev}})$ or $H$ contains all vertices from the cycle not containing $v_{\mathrm{cur}}$, the current rotation vertex (or both conditions hold).
\end{enumerate}

Before showing that the invariants \ref{j:1} -- \ref{j:4} indeed hold,
we first show how to obtain a plane monochromatic spanning subdrawing from the output of our algorithm under the assumption that \ref{j:1} -- \ref{j:4} are true.
Invariant \ref{j:4} guarantees the termination of our algorithm. 
Further, by \ref{j:1} and \ref{j:2}, it follows that the union of the result $H$ of the algorithm and the two cycles contains a monochromatic plane spanning subdrawing.  
Indeed, let $H$ be the output of our algorithm. 
Then $V(H)$ contains all vertices of the cycle that does not contain the last rotation vertex.
Assume first that this cycle is blue. 
As by \ref{j:2}, all blue vertices are incident to a red edge in $H$, the red cycle together with the red edges in $H$ forms a spanning subdrawing in $D$, which, by \ref{j:1}, is plane.
Analogously, if the cycle not containing the last rotation vertex is red, then $V(H)$ contains all vertices of the red cycle, each of which is incident to a blue edge by \ref{j:2}.
Hence, the blue cycle and the blue edges in $H$ form a plane spanning subdrawing of $D$.

\paragraph{Proving the invariants.}

Invariants \ref{j:1} and \ref{j:2} follow quite straightforwardly from the construction, whereas \ref{j:4} is more involved. 
Recall that in the rotation of any vertex $v$, all side edges incident to $v$ appear consecutively.
Moreover, we state the following observation, which will be useful for proving \ref{j:1} and \ref{j:4}.

\begin{obs}\label{obs:order}
	In the rotation of any vertex $v$, the order of edges to the vertices of each circle is the same as the order along that circle. In particular, if $v_1,\ldots,v_k$ are all vertices on the circle not containing $v$ in circular order, then there exists a $1\leq j \leq k$ such that $v_j,v_{j+1},\ldots v_k, v_1, \ldots v_{j-1}$ appear in that order in the rotation around $v$.
\end{obs}

\paragraph{{\bf (J1).}} Observation~\ref{obs:order} together with the fact that incident edges must not intersect and we stop as soon as we reach an edge incident to $v$ or $v'$ implies \ref{j:1} (remember, $v$ and $v'$ are the incident vertices of the very first edge of \ref{a:2}).

\paragraph{{\bf (J2).}} All leaves that are attached to the backbone path fulfill \ref{j:2} by construction. Concerning, the vertices on the backbone path, we only switch cycles when reaching an edge of color different than the current rotation vertex. Hence, all but the last rotation vertex fulfill \ref{j:2}.

\paragraph{{\bf (J3).}} Let $i\geq 2$ and consider iteration $i$ of \ref{a:2}. Let $H$ be the \caterpillar\ at the end of this current iteration and $H_{\text{prev}}$ the one at the end of the previous iteration. Let $v$ be the first and let $x$ be the last rotation vertex of the previous iteration, i.e., $x$ is the first rotation vertex of the current iteration. Let $z$ denote the last rotation vertex of the current iteration. 

Then, we need to show that $H$ covers all vertices from the cycle not containing~$z$ or $V(H)$ is a strict superset of $V(H_{\text{prev}})$. To this end, we consider the following cases depending on the relative position of $v$ and $x$. 

\begin{description}
	\item[Case 1: $\mathbf{v}$ and~$\mathbf{x}$ lie on the same cycle.]\
	
	 In this case (the first and the last rotation vertex lie on the same cycle) we argue that our algorithm in fact covered all vertices from the other cycle, i.e., is already finished before triggering a new iteration.  
	
	Without loss of generality, let $v$ and $x$ be blue vertices and assume there is a red vertex $y$ that has not been covered by $H_{\text{prev}}$. By Observation~\ref{obs:order}, this vertex must lie \enquote{behind} the already considered vertices on the red cycle. If $v$ is equal to $x$, i.e., we considered only a single rotation vertex, we covered all red vertices. Otherwise, the edges $vy$ and $xy$ must intersect (see \cref{fig:cylindrical_stuck_same}), which is not possible in a simple drawing.\\
	
	  \begin{figure}[htb]
		\centering
		\includegraphics[scale=0.6,page=2]{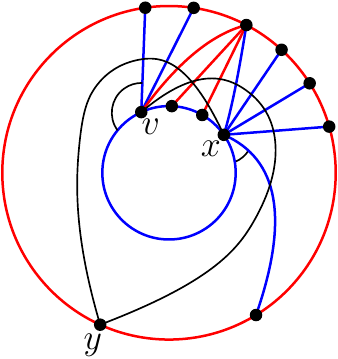}
		\caption{The black arcs around~$v$ and~$x$ indicate that there are no edges incident to~$x$ (resp.~$v$) in this direction. This forces the black edges $yv$ and $yx$ to intersect, which is forbidden in a simple drawing. }
		\label{fig:cylindrical_stuck_same}
	\end{figure}

	\item[Case 2: $\mathbf{v}$ and~$\mathbf{x}$ lie on different cycles.]\
	
	This is the more interesting case, that indeed triggers a new iteration of our algorithm in the reverse direction.
	
	Assume, without loss of generality, that $v$ is a blue vertex and $x$ is red (i.e., the previous iteration started on the blue cycle and the current iteration on the red cycle). The argument of Case~1 of course also applies to the current iteration and hence, we can safely assume $z$ to be a blue vertex.

Remember that every blue vertex $u \in V(H_{\text{prev}})$ (of the previous iteration) is incident to a red edge $uu_r \in E(H_{\text{prev}})$ (due to \ref{j:2}). Then, the following observation turns out to be very helpful.

\begin{claim}
When rotating around a blue vertex $u \in V(H_{\text{prev}})$ in the \emph{current} iteration at latest we switch cycles with the edge $uu_r$, i.e., it is not possible to \enquote{skip} this red edge of the previous iteration.
\end{claim}

\begin{claimproof}
Assume that this is not true and let $u$ be the \emph{first} blue rotation vertex violating this property, i.e., $u$ is incident (in $H$) to a red vertex \emph{after} $u_r$. Let $ru$ be the blue backbone edge (in $H$) that led from $r$ to $u$. In particular, $r$ lies behind $u_r$. This obviously also implies that $r$ is not equal to $x$. So, let $br$ be the red backbone edge (in $H$) that led from $b$ to $r$. In particular, the algorithm considered $b$ before $u$. Moreover, since the edges $uu_r$ and $br$ intersect, $r$ must be behind $b_r$ (the neighbor of $b$ incident to $b$'s red backbone edge in the previous iteration). Hence, $b$ must have skipped its red edge $bb_r$ from the previous iteration (see \figurename~\ref{fig:cylindrical_skip}). This is a contradiction to $u$ being the first such blue vertex. 
\end{claimproof}
		\vspace{2ex}

 \begin{figure}[htb]
		\centering
		\includegraphics[scale=0.6,page=3]{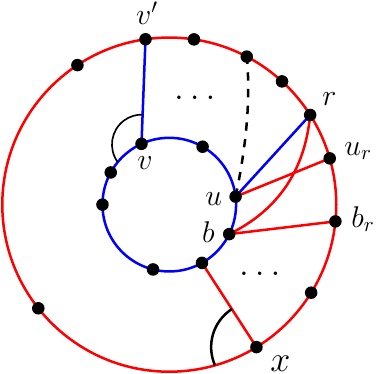}
		\caption{If the blue rotation vertex $u$ is incident to some edge behind $u_r$ in the iteration from $x$ to $v$ (the dashed edge), then $r$, the neighbor of the blue backbone edge, is also behind $u_r$. Hence, $b$ must also be a blue rotation vertex that skipped its red edge $bb_r$.}
		\label{fig:cylindrical_skip}
	\end{figure}	 

To summarize, $z$ is a blue vertex and by the above claim $z$ cannot
be in $V(H_{\text{prev}})$ (except if the stopping condition of covering all vertices from the other cycle was reached earlier).

Hence, it remains to show that all red vertices of $H_{\text{prev}}$ are also in $H$. If this was not the case, then in particular $v'$ is not in $H$ (Observation~\ref{obs:order}) and the edge $zv'$ would intersect the edge $vv'$ (see Figure~\ref{fig:cylindrical_revert}). Again, a contradiction to the drawing being simple.

    \begin{figure}[htb]
		\centering
		\includegraphics[scale=0.6,page=4]{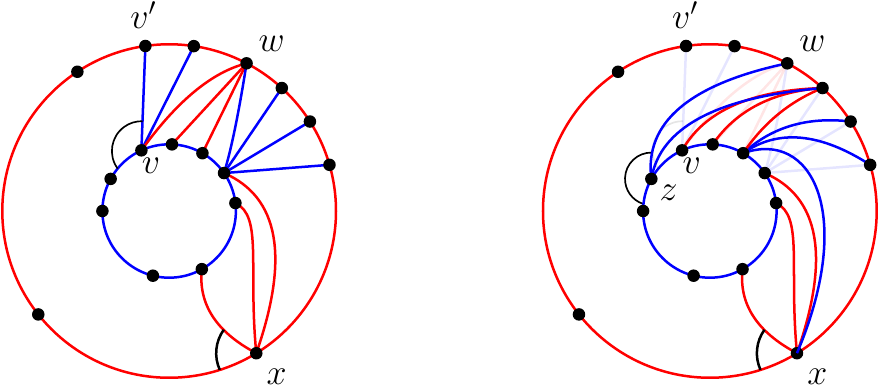}
		\caption{On the left, the algorithm started from $v$ and got stuck in $x$. On the right, the next iteration (in reverse direction) is illustrated. If we get stuck at the edge~$zw$ (rotating around $z$), there is no way to connect $z$ and $v'$ without crossing~$vv'$.}
		\label{fig:cylindrical_revert}
	\end{figure}

\end{description}

\paragraph{Step 2.} Now let $D$ be a 2-edge-colored cylindrical drawing that does not fulfill at least one of the properties \ref{p:1} and \ref{p:34}.

If it does not fulfill \ref{p:1}, the inner or outer cycle is empty, which implies that~$D$ is isomorphic to a 2-page book drawing and hence contains a monochromatic plane spanning tree (see Proposition~\ref{prop:2-page-book}).

So assume that $D$ fulfills \ref{p:1} but does not fulfill \ref{p:34}. 
If at least one of the cycles of $D$ is bichromatic (contains red and blue edges), then we iteratively remove a vertex whose incident cycle edges are of different color 
until we obtain a subdrawing $D'$ of $D$ in which both cycles are monochromatic. Clearly, $D'$ is a cylindrical drawing, since removing a vertex  cannot break any of the properties of a cylindrical drawing (all vertices still lie on the inner or outer circle, neither circle is crossed, and all edges between two vertices on the inner (outer) circle still lie completely inside (outside) that circle).

If the two cycles of $D'$ are of different color, $D'$ fulfills the properties \ref{p:1} and \ref{p:34} and hence contains a plane monochromatic spanning tree by Step~1.
If, on the other hand, the two cycles in $D'$ have the same color, then the union of them plus one side edge of that color gives a monochromatic plane spanning subdrawing for $D'$,
or, if such an edge does not exist, the according color class is not spanning and hence $D'$ contains a monochromatic plane spanning tree by \cref{obs:span}. 
Finally, as cycle edges are always uncrossed, 
we can extend the obtained spanning tree for $D'$ to one for~$D$ by re-adding the removed vertices in inverse order by \cref{obs:2_uncrossed}.~\hfill$\qed$
\end{proof}

\end{document}